\documentclass[12pt]{amsart}

\usepackage{fullpage}
\usepackage[all]{xy}
\usepackage[usenames,dvipsnames]{color}
\usepackage{longtable}
\usepackage{enumerate}
\usepackage{hyperref}
\usepackage{underscore}
\usepackage{cite}
\usepackage{bm}

\usepackage{amsmath} 
\usepackage{verbatim}
\usepackage{textcomp}
\usepackage{amsfonts}
\usepackage{amssymb}
\usepackage{amsthm}
\usepackage{fancyhdr}
\usepackage{graphicx}
\usepackage[english]{babel}
\usepackage{wrapfig}
 \usepackage{caption} \usepackage{subcaption} 
\usepackage[margin=1in]{geometry}
\usepackage{tikz} 
\usepackage[normalem]{ulem}
\usetikzlibrary{arrows}

\newcommand{\F}{\mathbb{F}}
\newcommand{\Z}{\mathbb{Z}}

\newcommand{\N}{\mathbb{N}}

\newcommand{\bpf}{\begin{proof}}
\newcommand{\epf}{\end{proof}}

\definecolor{darkviolet}{rgb}{0.58, 0.0, 0.83}
\newcommand{\rmv}[1]{}

\newcommand{\lub}{\operatorname{lub}}
\newtheorem{theorem}{Theorem}
\newtheorem{lemma}[theorem]{Lemma}
\newtheorem{proposition}[theorem]{Proposition}
\newtheorem{corollary}[theorem]{Corollary}

\newenvironment{example}[1][Example]{\begin{trivlist}
\item[\hskip \labelsep {\bfseries #1}]}{\end{trivlist}}
\newenvironment{remark}[1][Remark]{\begin{trivlist}
\item[\hskip \labelsep {\bfseries #1}]}{\end{trivlist}}

\DeclareMathOperator{\Hull}{Hull}
\DeclareMathOperator{\Supp}{Supp}
\DeclareMathOperator{\Div}{Div}
\DeclareMathOperator{\lmd}{lmd}

\begin{document}
\title{Explicit non-special divisors of small degree, algebraic geometric hulls, and LCD codes from Kummer extensions}

\author{Eduardo Camps Moreno}
\address[Eduardo Camps Moreno]{Department of Mathematics, Virginia Tech\\ Blacksburg, VA USA}
\email{e.camps@vt.edu}

\author{Hiram H. L\'opez}
\address[Hiram H. L\'opez]{Department of Mathematics, Virginia Tech\\ Blacksburg, VA USA}
\email{hhlopez@vt.edu}

\author{Gretchen L. Matthews}
\address[Gretchen L. Matthews]{Department of Mathematics\\ Virginia Tech\\ Blacksburg, VA USA}
\email{gmatthews@vt.edu}
\thanks{Eduardo Camps Moneno was partially supported by a CONACYT scholarship. Hiram H. L\'opez was partially supported by the NSF grants DMS-2201094 and DMS-2401558, and an AMS-Simons Travel Grant. Gretchen L. Matthews was supported by the NSF grants DMS-1855136 and DMS-2037833, and the Commonwealth Cyber Initiative. This work was performed while Eduardo Camps Moreno was at the Instituto Polit\'ecnico Nacional of Mexico and Hiram H. López was at Cleveland State University.}

\maketitle

\begin{abstract}
In this paper, we consider the hull of an algebraic geometry code, meaning the intersection of the code and its dual. We demonstrate how codes whose hulls are algebraic geometry codes may be defined using only rational places of Kummer extensions (and Hermitian function fields in particular). Our primary tool is explicitly constructing non-special divisors of degrees $g$ and $g-1$ on certain families of function fields with many rational places, accomplished by appealing to Weierstrass semigroups. We provide explicit algebraic geometry codes with hulls of specified dimensions, producing along the way linearly complementary dual algebraic geometric codes from the Hermitian function field (among others) using only rational places and an answer to an open question posed by Ballet and Le Brigand for particular function fields. These results complement earlier work by Mesnager, Tang, and Qi that use lower-genus function fields as well as instances using places of a higher degree from Hermitian function fields to construct linearly complementary dual (LCD) codes and that of Carlet, Mesnager, Tang, Qi, and Pellikaan to provide explicit algebraic geometry codes with the LCD property rather than obtaining codes via monomial equivalences. 
 \end{abstract}

\section{Introduction}

Non-special divisors play an important role in the design and many applications of algebraic geometry codes. Given an algebraic function field $F$ of genus $g$ over a finite field $\mathbb{F}$, a divisor $D$ on $F$ is non-special if and only if $\ell(D)=\deg D-g+1$, where $\ell(D)$ denotes the dimension of the Riemann-Roch space of $D$ as an $\F$-vector space. This property is usually guaranteed by requiring that $\deg D\geq 2g-1$. This condition helps to compute the dimension of an algebraic geometry code \cite[Corollary 2.2.3]{St} and to improve the Tsfasman-Vladut-Zink bound (see \cite[Chapter 5]{Adv} for a survey on these improvements). Low-degree non-special divisors have been used to build generalized algebraic geometry codes \cite{NXL}, to compute Weierstrass semigroups in towers of function fields \cite{tower}, and to improve bounds for the bilinear complexity of multiplication algorithms \cite{BB}. 
In 2006, Ballet and Le Brigand~\cite{BB} demonstrated the existence of non-special divisors of degree $g-1$ (resp., of degree $g$) for function fields of genus $g \geq 2$ over a field of size $q \geq 2$ (resp., $q \geq 3$) and noted the difficult task of identifying these divisors. 

In this work, we determine particular non-special divisors of low degree for the Hermitian function field and a more general Kummer extension, addressing the challenge posed in \cite{BB} and motivated by studying the hull of algebraic geometry codes. The hull of a linear code $C$ is $\Hull(C):=C \cap C^{\perp}$. Hulls were first studied formally in the early 1990s by Assmus and Key \cite{Assmus}, though their use may be tied to earlier notions. For instance, a code $C$ is self-orthogonal if and only if $\Hull(C)=C$. The dimensions of hulls govern the complexity of some algorithms of interest in cryptography  \cite{graph_isom}, \cite{Sendrier} as well as properties of some entanglement-assisted quantum error-correcting codes \cite{hull_quantum}. In addition, a code is linear complementary dual (LCD)  if and only if its hull is trivial. LCD codes, whose study was initiated by Massey \cite{Massey} in 1992, simultaneously protect against side-channel and fault injection attacks, as Carlet and Guilley \cite{CG} pointed out. Their connection to the specialty of divisors was observed by Mesnager, Tang, and Qi \cite{MTQ}. 

The determination of non-special divisors of low degrees is mathematically valuable. Additionally, an outgrowth is the construction of algebraic geometry with all of the following nontrivial properties: relying only on the well-understood rational places, the duals and hulls are supported in the same rational places, come from maximal function fields of higher genus than previous works, and utilize a smaller alphabet for certain lengths than prior constructions that also use rational places. Among these codes are LCD codes and other codes with small hulls that are algebraic geometric codes themselves. We distinguish this approach from that of Carlet, Mesnager, Tang, Qi, and Pellikaan \cite{linear}, which produces from a linear code over a finite field $\F_q$ with $q > 3$ elements an equivalent code that is LCD using linear algebraic transformations. The codes in this paper are traditional algebraic geometry codes whose design guarantees that they are LCD. 

In this paper, we construct algebraic geometry codes with hulls of specified dimensions using only rational places on certain maximal function fields, including the Hermitian function field. More generally, we consider Kummer extensions given by 
\begin{equation} \label{Kummer_eq}
\prod_{i=1}^r (y-\alpha_i)=x^m
\end{equation} 
over $\F_q$ with $(q,m)=1$ and $(r,m)=1$. These Kummer extensions allow us to consider several families of function fields of particular interest in coding theory, including the extended norm-trace curve given by 
$$
y^{q^{t-1}}+y^{q^{t-2}}+ \dots + y=x^u 
$$
over $\F_{q^t}$ where $u | \frac{q^{t}-1}{q-1}$. Hermitian function fields, which are maximal over $\F_{q^2}$, are seen by taking $t=2$ and $u=\frac{q^{t}-1}{q-1}$. The first examples of non-classical function fields are obtained when $t=2$ \cite{Schmidt} and are also maximal. The norm-trace function fields, obtained by setting $u=\frac{q^{t}-1}{q-1}$, meet the Geil-Matsumoto bound on the number of rational places over $\F_{q^t}$.

Section~\ref{prelim} shows that we must consider multipoint codes to achieve our goals; this section also provides the background on Weierstrass semigroups, which will be used to address the problem. Our primary tool is the explicit construction of non-special divisors of small degrees, as detailed in Section \ref{explicit_divisors_section}. We build on the work of Ballet and Le Brigand~\cite{BB}, where the authors show the existence of non-special divisors of degree $g-1$ (resp., of degree $g$) for function fields of genus $g \geq 2$ over a field of size $q \geq 2$ (resp., $q \geq 3$). Here, we provide explicit constructions for such divisors on the Kummer extensions given in Eq.~(\ref{Kummer_eq}). In Section \ref{codes_euclidean_section}, we use these non-special divisors to construct a family of codes for which the hull is known explicitly, including, but not at all limited to, some algebraic geometry codes from rational places, which are LCD. Section~\ref{conclusion} presents a conclusion.

\section{Preliminaries} \label{prelim}
We begin this section by setting up the notation and terminology used throughout the paper. We refer the reader to the standard references~\cite{St} and \cite{vL} for more details on algebraic geometry codes. 

Let $F$ be an algebraic function field of genus $g$ over a finite field $\F$ with algebraic closure $\overline{\F}$. The field of rational functions of $F$ is denoted by $\F\left( X \right)$, the set of $K$-rational places of $F$ is written $X(K)$ for an intermediate field $K$ of the extension $\overline{\F}/\F$, and the vector space of Weil differentials is given by $\Omega_X$. Given a divisor $A$, we may write $A=\sum_{P \in X(\overline{\F})} a_P P$ with $a_P \in \Z$ and all but finitely many $a_P=0$; in this case, we say $v_P(A):=a_P$ and $\deg A = \sum_{P \in X(\overline{\F})} v_P(A)$. There is a partial order on $\Div(X)$, the set of divisors of $F$, given by $A \leq B$ if and only if $v_P(A) \leq v_P(B)$ for all $P \in X(\overline{\F})$. The divisor of a function $f \in \F\left( X \right) \setminus \{ 0\}$ is written  as $(f)=\sum_{P \in X(\F)} a_P P$ where $P$ is a zero (resp., pole) of multiplicity $a_P$ (resp., $-a_P$) provided $a_P>0$ (resp., $a_P<0$). The pole divisor of $f \in \F\left( X \right) \setminus \{ 0\}$ is $(f)_{\infty} =\sum_{{\footnotesize{\begin{array}{l} P \in X(\F) \\ v_P(f)<0 \end{array}}}} a_P P$.
A divisor $A$ of $F$ defines a space of functions 
$$\mathcal{L}(A):= \left\{ f \in \F\left( X \right): (f) \geq - A \right\} \cup \{ 0 \}$$ of $F$ as well as a space of differentials 
$$\Omega(A):=\left\{ \omega \in \Omega_X: (\omega) \geq A \right\} \cup \{ 0 \},$$
where $(\omega) \in \Div(X)$ denotes the divisor associated with the differential $\omega$.
The dimension of $\mathcal L(A)$ is given by $\ell \left( A \right)$ and satisfies the Riemann-Roch Theorem, meaning $$\ell \left( A \right)= \deg A +1-g + \ell \left( K-A \right)$$ where $K$ is a canonical divisor of $F$. Moreover, if $\deg A \geq 2g-1$, then 
\begin{equation} \label{RR_high_deg}
\ell \left( A \right)= \deg A +1-g.
\end{equation} We will use the fact that $$A \leq B \Rightarrow \mathcal{L}(A) \subseteq \mathcal{L}(B).$$ The support of the divisor $A$ is  $\Supp(A) := \left\{ P \in X \left( \F \right) : v_P \left( A \right) \neq 0 \right\}$, and we say that $A$ is supported by $P \in  X \left( \F \right)$ if and only if $v_P \left( A \right) \neq 0$. A divisor $A$ is linearly equivalent to $B$, denoted $A \sim B$, if and only if there exists $f \in \F(X)$ such that $A-B=(f)$. 
In the setting where we consider a function field $X: f(y)=g(x)$,  we use $P_{ab}$ to denote a place of $F$ corresponding to $x=a$ and $y=b$. If $F$ has a unique place at infinity, we note it by $P_{\infty}$.

We use the standard notation from coding theory. Because this paper only considers linear codes, we use the term code to mean linear code. An $[n,k,d]$ code over a finite field $\F$ is a code of length $n$, dimension $k$, and minimum distance $d$ (taken with respect to the Hamming metric). The set of indices of codewords of a code of length $n$ is $[n]:=\left\{ 1, \dots, n \right\}$. Given ${v} \in \F^n$, we denote its $i^{th}$ component by $v_i$ where $i \in [n]$. The dual of an $[n,k,d]$ code $C$ is  $$C^{\perp}:= \left\{ w \in \F^n: w \cdot c=0 \ \forall c \in C \right\};$$ that is, the dual is taken with respect to the usual (or Euclidean) dot product. \rmv{For a code $C$ over $\F_{q^2}$, we may also consider its Hermitian dual $$C^{\perp_h}:=\left\{ w \in \F_{q^2}^n: w \cdot_h c = 0 \ \forall c \in C \right\}$$ where
$x \cdot_h y:=\sum_{i=1}^n x_i y_i^q$ is the Hermitian inner product of words $x, y \in \F_{q^2}^n$.} The set of $m \times n$ matrices over a field $\F$ is denoted $\F^{m \times n}$. We write $\N$ to mean the set of nonnegative integers and  $\Z^+$ for the set of positive integers. At times, we make use of the partial order $\leq$ on $\N^m$ given by $v \leq w$ if and only if $v_i \leq w_i$ for all $i \in [m]$; we write $v \nleq w$ to mean there exists $i \in [m]$ with $v_i > w_i$. 
 
Suppose $G$ and $D := Q_1 + \dots +Q_n$ are  divisors of $F$ defined over $\F$, 
where $Q_1, \dots , Q_n$ are distinct $\F$-rational places of $F$, each not belonging to the support of $G$.  
The {\it algebraic geometry code} $C(D,G)$ is defined by
$$
C(D,G) := \{ ev(f) : f \in \mathcal L(G) \},
$$
where
$$
ev(f):=(f(Q_1), f(Q_2),\dotsc, f(Q_n)).
$$
Certainly, $n$ is the length of the code. For convenience, we suppose $n > \deg G$, so the evaluation map $ev: \mathcal L(G) \rightarrow \F^n$ is injective. 
Then, the code $C(D, G)$  has
dimension $k=\ell(G)$ and minimum Hamming distance $d \geq n-\deg G$.
If the support of $G$ consists of $m$ places, then $C(D, G)$ is called an $m$-point code and said to be a {\it multipoint code} provided $m \geq 2$.  

According to \cite[Proposition 2.2.8]{St}, the dual of $C(D,G)^{\perp}$ is 
$$
C(D,G)^{\perp}= \{ \left( \omega_{Q_1}(1), \dots, \omega_{Q_n}(1) \right): \omega \in \Omega(G-D) \}.$$ It is well-known that if there is a differential $\eta \in \Omega_X$, which has simple poles at the places in the support of $D$, then 
\begin{equation} \label{dual_gen}
C(D,G)^{\perp} = \left\{ \left( \eta_{Q_1}(1) f(Q_1), \eta_{Q_2}(1)f(Q_2),\dotsc, \eta_{Q_n}(1)f(Q_n)\right) : f \in \mathcal L(D-G+(\eta)) \right\},
\end{equation}
meaning that the dual of the algebraic geometry code is a generalized algebraic geometry code. This leads to the following result (cf. \cite[Proposition 2.2.10]{St}).  

\begin{lemma}  \label{dual_AG}
Consider the algebraic geometry code $C(D, G)$ where $D=Q_1+\dots+Q_n$ as above.
    \begin{itemize}
\item[\rm (1)] If there exists $\eta \in \Omega_X$ such that $v_{Q_i}(\eta) \geq -1$ and $\eta_{Q_i}(1)=1$ for all $i \in [n]$, then $$C(D,G)^{\perp}=C(D,D-G+(\eta)).$$
\item[\rm (2)]  If, in addition, $G=\sum_{i=1}^m a_i P_i$ so that $C(D,G)$ is an $m$-point code and $\Supp (\eta) \subseteq \Supp (G) \cup \Supp (D)$, then $C(D,G)^{\perp}= C(D,\sum_{i=1}^m a_i' P_i)$ for some $a_i' \in \Z$.
\end{itemize}
\end{lemma}

Consider a one-point code $C(D, \alpha P)$ on a function field $F$ satisfying the hypotheses of Lemma~\ref{dual_AG}~(2). Then, the dual of the one-point code $C(D, \alpha P)$ of $F$ is also a one-point code, say $C(D, \alpha' P)$. In this case,  
\begin{equation} \label{1pt}
\begin{array}{lll}
C(D, \alpha P) \cap C(D, \alpha P)^{\perp} &=& C(D, \alpha P) \cap C(D, \alpha' P) \\
& =& C(D, \min \left\{ \alpha, \alpha' \right\} P),
\end{array}
\end{equation}
so the hull is well-understood. This intersection is trivial if and only if $\alpha<0$ or $\alpha'<0$, meaning there are no nontrivial LCD one-point codes of $F$.  For the multipoint case, the hull is not as simple, and in fact, we have that
$$C\left(D,\sum\min\{a_P,b_P\} P\right)\subseteq C\left(D,\sum a_PP\right)\cap C\left(D,\sum b_PP\right),$$
which implies that, in some cases, we need further information to get a complete hull description. Hence, it would be desirable to define divisors supported on a set $S$ of rational points, such that the corresponding algebraic geometry code is also of the same type and the parameters of such divisors can fully describe the hull. For the property on the support of the divisors, Lemma \ref{dual_AG}~(2) provides us with a sufficient condition. Many maximal function fields, including the Hermitian function field, satisfy the conditions of Lemma~\ref{dual_AG}~(2).

To define the desired codes, we focus on explicitly constructing non-special divisors. Recall that the {\it index of specialty} of a divisor $A$ is $$i(A):=\ell(A) - \left( \deg(A)+1-g \right).$$ A divisor $A$ is called {\it non-special} if and only if $i(A)=0$, meaning $\ell(A)=\deg(A)+1-g$, and special otherwise. Mesnager, Tang, and Qi \cite{MTQ} used non-special divisors to construct LCD codes. Based on their contribution and motivated by Eq.~(\ref{1pt}), we define two useful operators on divisors. Given divisors $G$ and $H$ on an algebraic function field $F$ over a finite field $\F$, define their {\it greatest common divisor} as 
$$
\gcd(G,H):= \sum_{P \in X \left(\overline{\F}\right)} \min \left\{ v_P \left( G \right), v_P \left( H \right) \right\} P
$$
and their {\it least multiple divisor} as 
$$
\lmd(G,H):= \sum_{P \in X \left(\overline{\F}\right)} \max \left\{ v_P \left( G \right), v_P \left( H \right) \right\} P.
$$
It is immediate that $$\mathcal L(G) \cap \mathcal L(H) = \mathcal L(\gcd(G,H)),$$ 
$$f \in \mathcal L(G) \text{ and } h \in 
\mathcal L(H) \Rightarrow  f-h \in \mathcal L(\lmd(G,H)),$$
and $$G+H=\gcd(G,H)+\lmd(G,H).$$ 
Using these ideas, we prove the following result, which is strongly inspired by \cite[Theorem 4]{MTQ}. We will later combine it with other tools to produce codes using only rational places on the Hermitian function field and some of its relatives, with the property that their hull is fully described as an algebraic geometry code, obtaining some LCD codes as a consequence.

\begin{theorem} \label{thm_2_10_25}
    Let $G$, $H$, and $D=\sum_{i=1}^n Q_i$ be divisors on a function field $F$ such that $Q_i\neq Q_j$ for $i\neq j$ and $G$ and $H$ each have disjoint support from $D.$ Assume that $G+H-D$ is a canonical divisor such that
    $$C(D,G)^\perp=C(D,H).$$
    If $\gcd(G,H)$ is non-special, then
    $$\Hull(C(D,G))=C(D,\gcd(G,H)).$$
\end{theorem}

\begin{proof}
It is clear that $C(D,\gcd(G,H))=ev(\mathcal{L}(G)\cap\mathcal{L}(H))\subseteq C(D,G)\cap C(D,H)$. On the other hand, if $c\in \Hull(C(D,G))$, then there exists $f\in\mathcal{L}(G)$ and $h\in\mathcal{L}(H)$ such that
$$ev(f)=ev(h)=c$$
and then $f-h\in\mathcal{L}(\lmd(G,H)-D)$. Since $G+H-D$ is canonical, $\ell(\lmd(G,H)-D)=i(\gcd(G,H))=0$ since $\gcd(G,H)$ is non-special. Then $f=h$ and $f\in\mathcal{L}(\gcd(G,H))$.
\end{proof}

\begin{remark}
While the proof of Theorem \ref{thm_2_10_25} makes use of the fact that $\gcd(G, H)$ is non-special, this is not a necessary condition to have $\Hull(C(D, G))=C(D,\gcd(G, H))$. Indeed, consider the divisors $G=(q^2-q-2)P_{\infty}$ and $H=q^3P_{\infty}$ on the Hermitian function field $F/\F_{q^2}$, where $D$ is the sum of all rational places of $F$ other than $P_{\infty}$. Then $G+H-D=\left(dx/(x^{q^2}-x) \right)$ and $C(D,G)^{\perp} = C(D,H)$. Therefore, $\Hull(C(D,G))=C(D,(q^2-q-2)P_{\infty}) \cap C(D,q^3P_{\infty}) = C(D,(q^2-q-2)P_{\infty})=C(D,\gcd(G,H))$. However, $\gcd(G,H)$ is special, because $i((q^2-q-2)P_{\infty}))= \frac{q(q-1)}{2} - \left( q^2-q-2 + 1 - \frac{q(q-1)}{2} \right) =1 > 0$. 
\end{remark}

\begin{corollary} \label{LCD_cor}
    Let $G$,$H$, and $D$ divisors as in Theorem~\ref{thm_2_10_25}.
    \begin{enumerate}
        \item The dimension of the hull of $C(D,G)$ is $\ell(\gcd(G,H))$.
        
        \item If $\deg\gcd(G,H)=g-1$, then $C(D,G)$ is LCD.
        
        \item If $G\leq H$, then $C(D,G)$ is self-orthogonal. If equality holds, then the code is self-dual.
    \end{enumerate}
\end{corollary}

In light of Corollary \ref{LCD_cor}, we set out to find explicit non-special divisors of degree $g-1$ since, as we will see, this is the lowest degree of a non-special divisor and, correspondingly, describes codes with the smallest hull. A primary tool in our approach is the Weierstrass semigroup. Given $m$ distinct $\F$-rational places $P_1, \dots, P_m$ on a function field $F$, the {\it Weierstrass semigroup} of the $m$-tuple $(P_1, \dots, P_m)$ is
$$
H(P_1, \dots, P_m) := \left\{ {\alpha} \in \N^m :
\exists f \in \F(X) \text{ with } (f)_{\infty}= \sum_{i=1}^{m} \alpha_i P_i \right\}.
$$
Equivalently, $\mathbf{\alpha} \in H(P_1, \dots, P_m)$ if and only if $$\ell \left( \sum_{i=1}^m \alpha_i P_i \right) = \ell \left( \sum_{i=1}^m \alpha_i P_i - P_j \right)+1$$ for all $j \in [m]$.  For short, if the $m$-tuple of places $(P_1, \dots, P_m)$  is clear from the context, we sometimes write $H_m$ to mean $H(P_1, \dots, P_m)$. The Weierstrass gap set of a $t$-tuple of places $(P_1, \dots, P_m)$ is $$G(P_1, \dots, P_m):=\N^m \setminus H(P_1, \dots, P_m).$$ 

Notice that $\alpha \in H(P_1, \dots, P_m) \setminus \{ 0\}$ implies $\ell \left( \sum_{i=1}^m \alpha_i P_i  \right) \geq 2$. Indeed, there exists $j \in [m]$ with $\alpha_j > 0$, and 
$$\ell \left( \sum_{i=1}^m \alpha_i P_i \right) > \ell \left( \sum_{i=1}^m \alpha_i P_i - P_j \right) \geq \ell \left( 0 \right) = 1 $$
since $\sum_{i=1}^m \alpha_i P_i - P_j  \geq 0$. 
The multiplicity of the semigroup $H(P)$, where $P \in X(\F)$, is $$\gamma (H(P)):= \min \left\{ a: a \in H(P) \setminus \{ 0\} \right\}.$$ In the next section, we use features of these semigroups to define non-special divisors of the least possible degree.

\section{Explicit non-special divisors of small degree} \label{explicit_divisors_section}

Recall that  for a divisor $A \in \Div(X)$ with $\deg A \geq 2g-1$, $i(A)=0$ according to Eq.~(\ref{RR_high_deg}). Hence, divisors of degree at least $2g-1$ are non-special.  On the other hand, if $\deg A< g-1$, then $\deg A + 1 -g <0$ and so $\ell(A) \neq  \deg A + 1 -g$. Thus, divisors of degrees less than $g-1$ are special. Consequently, the least possible degree of a non-special divisor is $g-1$. It is also worth noting that if $A \in \Div(X)$ has $\deg A = g-1$, then $A$ non-special implies $A$ is not effective. To see this, notice that if $A \geq 0$, then $\F \subseteq \mathcal{L}(A)$; hence, $\ell(A) \geq 1 \neq 0 =  \deg(A)+1-g$. We conclude that the least possible degree of an effective non-special divisor is $g$. This is captured in Figure \ref{line}. Furthermore, if $A$ is non-special and $A \leq B$, then $B$ is also non-special. This follows from the fact that if $\ell(A)=\deg A + 1 -g$, then $\ell(K-A)=0$ for a canonical divisor $K$ on a function field $F$, which implies $\ell(K-B)=0$ and $\ell(B)=\deg B + 1 -g.$ 
Thus, if we have $C(D,G)^\perp=C(D,H)$ and $A\leq G,H$ for some non-special divisor of degree $g-1$, we will have that $C(D,G)\cap C(D,H)=C(D,\mathrm{gcd}(G,H))$. 

\begin{figure} 
\begin{tikzpicture}[xscale = 1.5]
\draw [thick, <->] (-4.5,0) -- (5.0,0);
\draw (0,-0.2) -- (0,0.2); 
\draw (-1,-0.2) -- (-1,0.2); 
\draw (2.5,-0.2) -- (2.5,0.2); 
\draw (-2,-0.2) -- (-2,0.2); 
\node [above] at (0,0.2) {$g$};
\node [below] at (0.2,-1.65) {least possible degree of};
\node [below] at (0.2,-2.1) {effective non-special $A$};
\node [above] at (-1,0.2) {$g-1$};
\node [above] at (2.5,0.2) {$2g-1$};
\node [above] at (4.5,0.2) {$\deg A$};
\node [above] at (-2,0.2) {$g-2$};
\node [below] at (-1,-.65) {$A$ special if};
\node [below] at (-1,-1.1) {effective};
\node [below] at (-3,-.65) {$A$ special};
\node [below] at (3.7,-.65) {$A$ non-special};
\draw [red,  thick, -stealth] (-2,-0.2) -- ++(0,-0.5) -- ++(-2.5,0);
\draw [blue, thick, -stealth] (2.5,-0.2) -- ++(0,-0.5) -- ++( 2.5,0);
\draw [darkviolet, thick, ->] (-1,-.75) -- (-1,-.25);
\draw [darkviolet, thick, ->] (0,-1.7) -- (0,-.25);
\end{tikzpicture}
\caption{Impact of $\deg A$ on the index of speciality of a divisor $A$ on a function field $F$ of genus $g$} \label{line}
\end{figure}
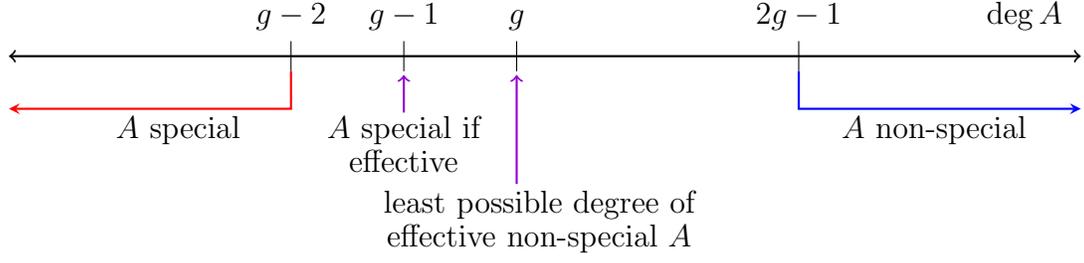

Ballet and Le Brigand \cite{BB} proved the existence of non-special divisors of degree $g-1$ (resp., of degree $g$) for function fields of genus $g \geq 2$ over fields of size $q \geq 2$ (resp., $q \geq 3$). In particular, they show the following.

\begin{lemma} \label{BB_lemma}
Suppose $F$ is a function field of genus $g \geq 1$ over a field $\F$ with $| \F | \geq 4$. 
\begin{enumerate}
\item If $|X(\F)| \geq g$, then there exists a non-special divisor $A \geq 0$ with $\deg A = g$ and $\Supp  (A)  \subseteq X(\F)$.
\item If  $|X(\F)| \geq g+1$, then there exists a non-special divisor $A$ with $\deg A = g-1$ and $\Supp (A)  \subseteq X(\F)$.
\end{enumerate}
\end{lemma}

In this section, we give explicit constructions of non-special divisors of degrees $g-1$ and $g$ for certain families of function fields. We first focus our attention on those of degree $g$. To determine a non-special divisor $A$ of degree $g$, we seek $A$ with $\ell \left( A \right)=1$, meaning $\mathcal{L}(A)=\F$. 

Notice that if $A$ and $B$ are effective divisors with $A \leq B$ and $v_P(A) \in H(P) \setminus \{ 0 \}$, then $\ell \left( B \right) \geq \ell \left( A \right) \geq 2$.  
Since we seek $A=\sum_{i=1}^m \alpha_i P_i $ with $\ell \left( A \right)=1$, we may restrict our search to divisors $$A=\sum_{i=1}^m \alpha_i P_i  
\textnormal{ with } \alpha_i <  \gamma \left( H(P_i)\right) \ \forall i \in [m].$$ This alone, however, is not enough to guarantee $\ell \left( A \right)=1$ as the following example demonstrates.

\rmv{
\begin{example}
Consider the Hermitian function field $F$ given by $y^8+y=x^9$ over $\F_{64}$. It is well-known that $H(P)=\left< 8, 9 \right>$ for all $\F_{64}$-rational places $P$ of $F$, so $\gamma \left( H(P)\right)=8$ for all $P \in X(\F_{64})$. Furthermore, 
$$(x)=\sum_{{\footnotesize{\begin{array}{l} b \in \F_{64} \\ b^8+b=0 \end{array}}}} P_{0b}-8P_{\infty},$$ and $$(y)=9P_{00}-9P_{\infty}.$$ Consequently, 
$$
\left( \frac{x^2}{y}  \right) = 2 \sum_{{\footnotesize{\begin{array}{l} b \in \F_{64}^* \\ b^8+b=0 \end{array}}}} P_{0b}- 7 P_{00} -7P_{\infty}.
$$
It follows that $ \frac{x^2}{y} \in \mathcal{L}\left( 7P_{00}+7P_{\infty} \right)$ and $\ell \left( 7P_{00}+7P_{\infty} \right) \geq 2$ despite the fact that $$v_{P} \left( 7P_{00}+7P_{\infty}\right)= 7 <\gamma \left( H(P)\right)$$ for $P \in \left\{ P_{00}, P_{\infty} \right\}$. 
\end{example}}

\begin{example}
Consider the Hermitian function field $F$ given by $y^q+y=x^{q+1}$ over $\F_{q^2}$. It is well-known that $H(P)=\left< q, q+1 \right>$ for all $\F_{q^2}$-rational places $P$ of $F$, so $\gamma \left( H(P)\right)=q$ for all $P \in X(\F_{q^2})$. Furthermore, 
$$(x)=\sum_{{\footnotesize{\begin{array}{l} b \in \F_{q^2} \\ b^q+b=0 \end{array}}}} P_{0b}-qP_{\infty},$$ and $$(y)=(q+1)P_{00}-(q+1)P_{\infty}.$$ Consequently, 
$$
\left( \frac{x^2}{y}  \right) = 2 \sum_{{\footnotesize{\begin{array}{l} b \in \F_{q^2}^* \\ b^q+b=0 \end{array}}}} P_{0b}- (q-1) P_{00} -(q-1)P_{\infty}.
$$
It follows that $ \frac{x^2}{y} \in \mathcal{L}\left( (q-1)P_{00}+(q-1)P_{\infty} \right)$ and $\ell \left( (q-1)P_{00}+(q-1)P_{\infty} \right) \geq 2$ despite the fact that $$v_{P} \left( (q-1)P_{00}+(q-1)P_{\infty}\right)= q-1 <\gamma \left( H(P)\right)$$ for $P \in \left\{ P_{00}, P_{\infty} \right\}$. 
\end{example}

As we will see, it is necessary to consider divisors $A=\sum_{i=1}^m \alpha_i P_i $ where $$\beta \in G(P_1, \dots, P_m) \  \forall  \left( \beta_1, \dots, \beta_m \right) \leq \left( \alpha_1, \dots, \alpha_m \right).$$
This naturally leads one to the minimal generating set of the Weierstrass semigroup. For $i \in [m]$, let $\Gamma^+(P_i) := H(P_i)$, and for $\ell \geq 2$, let
\[
\Gamma^+(P_{i_1}, \ldots, P_{i_\ell}) :=
\left\{
v \in {\Z^+}^\ell :
 v \text{ is minimal in }
 \{w \in H(P_{i_1}, \ldots, P_{i_\ell}) : v_i = w_i\}
 \text{ for some }i \in [\ell] 
\right\}.
\]
The \textit{least upper bound} of $v^{(1)}, \ldots, v^{(t)} \in \mathbb{N}^n$ is given by 
\[
\lub(v^{(1)}, \ldots, v^{(t)}) := (\max\{ v^{(1)}_1,\ldots, v^{(t)}_1\},\hdots, \max\{v^{(1)}_n,\ldots,v^{(t)}_n\}) .
\]
For each $I \subseteq [m]$ let $\iota_I$ denote the natural inclusion $\mathbb{N}^{\ell} \to \mathbb{N}^m$ into the coordinates indexed by $I$.
The \textit{minimal generating set} of $H(P_1,\hdots, P_m)$ is 
\[
\Gamma(P_1,\hdots, P_m) 
:= \bigcup_{\ell=1}^m \bigcup_{\substack{I=\{i_1,\hdots,i_\ell\} \\ i_1 < \cdots < i_\ell}} \iota_I(\Gamma^+(P_{i_1},\hdots,P_{i_\ell})).
\]
The Weierstrass semigroup $H(P_1,\hdots, P_m)$ is completely determined by the minimal genera{\-}ting set $\Gamma(P_1,\hdots,P_m)$ \cite[Theorem 7]{Matthews}:
If $1\leq m < |\F|$, then
\begin{equation} \label{hlub}
H(P_1,\hdots,P_m) = \{\lub \{ {{v_1}}, \dots, {{v_m}} \}:  {{v_1}}, \dots, {{v_m}} \in\Gamma(P_1,\hdots,P_m)\}.
\end{equation}
The following observation relates the minimal generating set of the semigroup of the $m$-tuple $(P_1,\hdots, P_m)$ to non-special effective divisors of degree $g$.

\begin{proposition}\label{gns}
Consider  an effective divisor $A=\sum_{i=1}^m \alpha_i P_i \in \Div(X)$ of degree $g$. If ${{\gamma}} \nleq \alpha$ for all  ${{\gamma}} \in \Gamma(P_1, \dots, P_m)$, then $A$ is non-special. 
\end{proposition}

\begin{proof}
Let $A=\sum_{i=1}^m \alpha_i P_i \in \Div(X)$ where $\alpha_i \geq 0$ for all $i, 1 \leq i \leq m$ and $\sum_{i=1}^m \alpha_i=g$. We prove the contrapositive of the statement. Suppose $A$ is special. Then $\ell(A) > \deg A  + 1 -g =1$. Since $\ell(A) \geq 2$, there exists $w \in H(P_1, \dots, P_m)$ such that $w \leq \alpha $. According to Eq.~(\ref{hlub}), there exist  ${{v_1}}, \dots, {{v_m}}  \in  \Gamma(P_1, \dots, P_m)$ satisfying ${{v_1}} \leq  
\lub \{ {{v_1}}, \dots, {{v_m}} \}=w \leq \alpha$, completing the proof.
\end{proof}

We are interested in Kummer extensions defined by $$\prod_{i=1}^r (y-\alpha_i)=x^m$$ where $\alpha_i\in\mathbb{F}_q$, $(q,m)=1$, and $(r,m)=1$, which is a function field of genus 
$\frac{(m-1)(r-1)}{2}$. They allow us to consider several families of function fields of particular interest in coding theory, including those
defined by  $$y^{q^{t-1}}+y^{q^{t-2}}+ \dots + y=x^u$$ over $\F_{q^t}$
where $u | \frac{q^{t}-1}{q-1}$. This family contains some maximal function fields, such as the Hermitian function fields, seen by taking $t=2$ and $u=\frac{q^{t}-1}{q-1}$, as well as their quotients which are given in the case $t=2$. The norm-trace function fields, obtained by setting $u=\frac{q^{t}-1}{q-1}$, are not maximal unless $t=2$ but do meet the Geil-Matsumoto bound. 

To prove our result, we need an explicit description of $\Gamma^+(P_1,\ldots,P_l)$ where $P_i\in\Supp((x))\setminus P_\infty$. Similar results for the Hermitian and norm-trace function fields were obtained previously at \cite{Matthews} and \cite{Peachey}.

\begin{proposition}[{\cite[Theorem 3.2]{oneandtwo}, \cite[Theorem 10]{kummer}}]\label{theorem.gamma}
Let $F/\mathbb{F}_q(y)$ be the Kummer extension defined by
$$x^m=\prod_{i=1}^r (y-\alpha_i)$$
as above, and let $P_i$ be the place associated with $y-\alpha_i$. Then $$\Gamma^+(P_1)=\mathbb{N}_0\setminus\left\{mk+j\ \left|\ 1\leq j\leq m-1-\left\lfloor\frac{m}{r}\right\rfloor,\ 0\leq k\leq r-2-\left\lfloor\frac{rj}{m}\right\rfloor\right.\right\}$$
and for $2\leq l\leq r-\left\lfloor\frac{r}{m}\right\rfloor$, $\Gamma^+(P_1,\ldots,P_l)$ is given by
$$\left\{(ms_1+j,\ldots,ms_l+j)\ \left|\ 1\leq j\leq m-1-\left\lfloor\frac{m}{r}\right\rfloor,\ s_i\geq 0,\ \sum_{i=1}^l s_i=r-l-\left\lfloor\frac{rj}{m}\right\rfloor\right.\right\}.$$

\end{proposition}

Before applying this proposition to the construction of non-special divisors, we state the facts needed to determine their degrees. 

\begin{lemma}\label{lemma.aux1}
Let $r,m\in\mathbb{Z}^+$ be relatively prime and $t$ be an integer such that $0 \leq t < m$ and $t \equiv r\mod{m}$.
 For $1\leq j\leq m-1$,
$$\left\lfloor\frac{r(j+1)}{m}\right\rfloor-\left\lfloor\frac{rj}{m}\right\rfloor=\begin{cases}\left\lfloor\frac{r}{m}\right\rfloor+1 & j=\left\lfloor\frac{km}{t}\right\rfloor,\ 1\leq k\leq t-1\\ \left\lfloor\frac{r}{m}\right\rfloor&\text{otherwise.}\end{cases}$$
Furthermore,    $$\sum_{k=1}^{t-1}\left\lfloor\frac{km}{t}\right\rfloor=\frac{(m-1)(t-1)}{2}.$$
\end{lemma}

\begin{proof}
   First, observe that
    $$\left\lfloor\frac{r(j+1)}{m}\right\rfloor-\left\lfloor\frac{rj}{m}\right\rfloor=\left\lfloor\frac{r}{m}\right\rfloor+\left\lfloor\frac{t(j+1)}{m}\right\rfloor-\left\lfloor\frac{tj}{m}\right\rfloor.$$
 Let $\left\lfloor\frac{km}{t}\right\rfloor<j<j+1\leq\left\lfloor\frac{(k+1)m}{t}\right\rfloor$ for some $1\leq k\leq t-1$. Since $(r,m)=1$, $(t,m)=1$ and we can guarantee that
    $$k\leq \frac{tj}{m}<\frac{t(j+1)}{m}<k+1.$$
 Thus, $\left\lfloor\frac{tj}{m}\right\rfloor=\left\lfloor\frac{t(j+1)}{m}\right\rfloor$. On the other hand, if $j=\left\lfloor\frac{km}{t}\right\rfloor$ for some $1\leq k\leq t-1$, then $j+1=\left\lceil\frac{km}{t}\right\rceil$. Since $(t,m)=1$, 
    $$j<\frac{km}{t}<j+1\Longrightarrow \frac{tj}{m}<k<\frac{t(j+1)}{m},$$
 from which $\left\lfloor\frac{t(j+1)}{m}\right\rfloor-\left\lfloor\frac{tj}{m}\right\rfloor=1$. We can conclude the first result by the observation at the beginning of this proof.
    
From Proposition \ref{theorem.gamma}, we know that for the Kummer extension given by $\prod_{i=1}^m (y-\alpha_i)=x^t$, the number of gaps of $P_1$ is
\begin{eqnarray*}
\left|\left\{tk+j\ \left|\ 1\leq j\leq t-1,\ 0\leq k\leq m-2-\left\lfloor\frac{mj}{t}\right\rfloor\right.\right\}\right|
&=& \sum_{j=1}^{t-1}\left(m-1-\left\lfloor\frac{mj}{t}\right\rfloor\right) \\
&=& \frac{(m-1)(t-1)}{2}
\end{eqnarray*}
from which the conclusion follows.
\end{proof}

We are ready to provide explicit effective non-special divisors of degree $g$ over some Kummer extensions. Here, we will make use of particular places of the function field denoted $P_{0b_j}$ where $P_{0b_i} \neq P_{0b_j}$ unless $i=j$.

\begin{theorem}\label{kummer rgtm}
    Let $F/\mathbb{F}_q(y)$ be a Kummer extension given by
    $$\prod_{i=1}^r (y-\alpha_i)=x^m$$
  where $\alpha_i\in\mathbb{F}_q$, $(q,m)=1$, and $(r,m)=1$. For $1\leq j\leq m-1-\left\lfloor\frac{m}{r}\right\rfloor$, define the following values:
    \begin{itemize}
        \item $l_j=r-\left\lfloor\frac{rj}{m}\right\rfloor$.
        \item $s_j=l_j-l_{j+1}$ if $j<m-1-\left\lfloor\frac{m}{r}\right\rfloor$ and $s_{m-1-\left\lfloor\frac{m}{r}\right\rfloor}=l_{m-1-\left\lfloor\frac{m}{r}\right\rfloor}-1=\max\left\{1,\left\lfloor\frac{r}{m}\right\rfloor\right\}$.
    \end{itemize}
    Then a divisor $A$ is effective and non-special of degree $g$ with support contained in the set $\left\{ P_{0b} : \prod_{i=1}^r (b-\alpha_i)=0\right\}$  if and only if
     $$A=\sum_{j=1}^{m-1-\left\lfloor\frac{m}{r}\right\rfloor} j\sum_{i=1}^{s_j}  P_{0b_{j_i}}.$$
     In particular, if $r<m$,
     
    $$A=\sum_{j=1}^{r-1} \left\lfloor\frac{jm}{r}\right\rfloor P_{0b_j}.$$
\end{theorem}

\begin{proof}

First, note that $$A=\sum_{j=1}^{m-1-\left\lfloor\frac{m}{r}\right\rfloor} jD_j = \sum_{j=1}^{m-1-\left\lfloor\frac{m}{r}\right\rfloor} j  \sum_{i=1}^{s_j} P_{0b_{j_i}} $$
where $$D_j = \begin{cases} \sum_{i=1}^{s_j} P_{0b_{j_i}} & \textnormal{if } s_j > 0 \\
    0 & \textnormal{if } s_j =0 \end{cases}$$ \noindent for $1\leq j\leq m-1-\left\lfloor\frac{m}{r}\right\rfloor$. 

We will prove that $A$ is of degree $g$. Let $t=r\mod m$, we have:
\begin{align*}
    \deg A&=\sum_{i=1}^{m-1-\left\lfloor\frac{m}{r}\right\rfloor} js_j\\
    &=\sum_{i=1}^{m-1}j\left\lfloor\frac{r}{m}\right\rfloor+\sum_{k=1}^{t-1}\left\lfloor\frac{km}{t}\right\rfloor\\
    &=\frac{(m-1)m}{2}\left\lfloor\frac{r}{m}\right\rfloor+\frac{(m-1)(t-1)}{2}\\
    &=\frac{(m-1)}{2}\left(m\left\lfloor\frac{r}{m}\right\rfloor+t-1\right)\\
    &=\frac{(m-1)(r-1)}{2}=g
\end{align*}

\noindent where the second equality follows from Lemma \ref{lemma.aux1} (1) and the third one from Lemma \ref{lemma.aux1} (2). Therefore, $A$ is effective of degree $g$. We will prove now that $A$ is a non-special divisor using Proposition~\ref{gns}.

Take $v\in\mathbb{N}^{l_1-1}$ such that $A=\sum_{i=1}^{l_1-1} v_iP_i$. Since $v_i\leq m-1-\left\lfloor\frac{m}{r}\right\rfloor$ for all $i$, then by Proposition \ref{theorem.gamma} $v_i<w$ for any $w\in \Gamma^+(P_i)$, and so
$$\iota_{\{i\}}(w)\not\leq v$$

\noindent for any $w\in\Gamma^+(P_i)$. Take $w\in\Gamma^+(P_{i_j}\ |\ j\in I\subset[l_1-1])$. If for some $i$, $w_i>m-\left\lfloor\frac{m}{r}\right\rfloor$, then $w\not\leq v$, so assume $w=(k,\ldots,k)$ for some $1\leq k\leq m-1-\left\lfloor\frac{m}{r}\right\rfloor$. By Proposition \ref{theorem.gamma}, we know that $|I|=l_k$ and the number of entries of $v$ greater than or equal to $k$ is $\sum_{i=k}^{m-1-\left\lfloor\frac{m}{r}\right\rfloor} s_i=l_k-1$, therefore $\iota_I(w)\not\leq v$ for any $I$ of cardinality $l_k$. This demonstrates that $w\not\leq v$ for all $w\in\Gamma(\Supp( A))$ and therefore $A$ is non-special.

Set $\gamma=m-1-\left\lfloor\frac{m}{r}\right\rfloor$ and choose $B$ to be an effective non-special divisor of degree $g$ supported on $\Supp((x))$. If $v_{P}(B)\geq\gamma+1$ for some $P$, then $\iota_{i}(\gamma+1)\leq B$, contradicting the non-speciality of $B$.
    
            Write $B=\sum_{j=1}^{\gamma} jD_j$ where $D_j$ is zero or is the sum of distinct rational places of degree $1$ and $\Supp(D_j)\cap\Supp(D_h)=\emptyset$ for $j\neq h$. 
    
    Observe that $|\Supp(B)|\leq l_1-1<r=|\Supp((x)_0)|$. For $D_{\gamma}$ we know
    $$\deg D_{\gamma}\leq l_{\gamma}-1=s_{\gamma}$$
    
    \noindent otherwise it would contradict the non-specialty of $B$. Take $D'_\gamma\geq D_\gamma$, $\Supp(D'_\gamma)\subset\Supp(D_\gamma)\cup(\Supp((x)_0)\setminus\Supp(B))$, such that 
    $$\deg D'_\gamma=|\Supp (D'_\gamma)|=l_\gamma-1.$$
    
    Similarly, for $1\leq h<\gamma$, we know 
    $$\sum_{j=h}^\gamma \deg D_j\leq \deg D_h+\sum_{j=h+1}^\gamma \deg D'_j\leq l_h-1,$$
    
    \noindent so take $D'_h\geq D_h$, $\Supp(D_h)\subseteq\Supp((x)_0)\setminus\Supp (B+\sum_{j=h+1}^\gamma D'_j)$ such that $$\sum_{j=h}^\gamma \deg D'_j=\sum_{j=h}^\gamma |\Supp (D'_j)|=l_h-1.$$ From this construction, it is clear that 
    $$\deg D'_h=s_h$$
    
    \noindent and so
    $$g=\deg B\leq \sum_{j=1}^\gamma j\deg D'_j=\sum_{j=1}^\gamma js_j=g.$$
    
    Then $D'_h=D_h$ for any $1\leq h\leq \gamma$ and $B$ has the desired form.
    
    In the case $r<m$, by Lemma \ref{lemma.aux1} we have $s_j=1$ if $j=\left\lfloor\frac{km}{r}\right\rfloor$ and $0$ otherwise and then $D_j=P_k$ or $D_j=0$. 
\end{proof}

\begin{corollary}
    On the norm-trace function field given by $y^{q^{t-1}}+y^{q^{t-2}}+ \dots + y=x^{\frac{q^{t}-1}{q-1}}$ over $\mathbb{F}_{q^t}$, any effective non-special divisor of degree $g$ supported by places $P_{0b}$ is of the form
   $$\sum_{i \in S} iP_{0b_i}$$ where $S =  \left\{ 1, \dots, \frac{q^t-1}{q-1}-2 \right\} \setminus \left\{ i  :  i \equiv 0 \mod q \right\}$. 
   \end{corollary}

\begin{proof}
    Take $u=\frac{q^t-1}{q-1}$. Given that
    $$\left\lfloor\frac{(t-1)u}{q^{t-1}}\right\rfloor=u-2,$$
    
    \noindent it is enough to prove that $\left\lfloor\frac{jm}{t}\right\rfloor$ cannot be divisible by $q$, since
    $$u-2-|\{i\in[u-2]\ |\ q|i\}|=u-2-\frac{u-1}{q}+1=\frac{u-1}{q}(q-1)=q^{t-1}-1$$
    
    \noindent and then $A$ should be $\sum_{j=1}^{q^{t-1}-1}\left\lfloor\frac{ju}{q^{t-1}}\right\rfloor P_j$, implying that $A$ is non-special of degree $g$. By Theorem \ref{kummer rgtm}, any other divisor with these characteristics should be of this form.
    
  Thus, we will prove $q\not{|}\left\lfloor\frac{ju}{q^{t-1}}\right\rfloor$ for any $1\leq j\leq q^{t-1}-1$. If $j$ is such that $\left\lfloor\frac{ju}{q^{t-1}}\right\rfloor=qk$ for some $1\leq k\leq\frac{u-1}{q}-1$, then
    $$ju=q^tk+z=u(q-1)k+k+z.$$
        This expression implies that $u|k+z$, but
    $$k+z<\frac{u-2}{q}+q^{t-1}=\frac{uq-1}{q}<u.$$ 
    Then, no such $k$ exists, and the conclusion follows.
\end{proof}

\begin{corollary}
On the Hermitian function field $y^q+y=x^{q+1}$ over $\F_{q^2}$, any effective non-special divisor of degree $g$ with support contained in $\{P_{0b_i}\ :\ 1\leq i\leq q\}$ is of the form $A=\sum_{i=1}^{q-1} i P_{0b_i}$. 
\end{corollary}

\begin{remark}\rm
    Theorem \ref{kummer rgtm} also gives effective non-special divisors of degree $g$ for function fields of the form $X:y^{q^{t-1}}+\cdots+y=x^u$ over $\mathbb{F}_{q^t}$ where $u|\frac{q^t-1}{q-1}$.\end{remark}

Now that we have explicit constructions for non-special divisors of degree $g$, we use the following idea to obtain non-special divisors of degree $g-1$.  These divisors of degree $g-1$ will support the construction of LCD codes in Section~\ref{codes_euclidean_section}.

\begin{lemma} \cite[Lemma 3]{BB}\label{lemma.BB}
If $A$ is a non-special divisor of degree $g$ on a function field $F$ and there exists $P \in X \left( \F \right) \setminus \Supp (A)$, then $A-P$ is non-special. 
\end{lemma}

Lemma \ref{lemma.BB} and Theorem \ref{kummer rgtm} yield the following result. 

\begin{theorem} \label{deg_g-1} Let $F/\mathbb{F}_q(y)$ be a Kummer extension defined by 
$$\prod_{i=1}^r(y-\alpha_i)=x^m$$

\noindent with $\alpha_i\in\mathbb{F}_q$, $(q,m)=1$, and $(r,m)=1$. Then
  $$A= \left( \sum_{j=1}^{m-1-\left\lfloor\frac{m}{r}\right\rfloor} j\sum_{i=1}^{s_j}  P_{0b_{j_i}} \right) - P$$
 is a non-special divisor of degree $g-1$ for all $P\in\{P_{ab}\ |\ a\neq 0\text{ or }b\neq b_{j_i}\}\cup\{P_\infty\}$.
In particular, there exist non-special divisors of degree $g-1$ supported on $\Supp((x)_0) \cup \left\{ P_{ab} \right\}$ for any $a \neq 0$. 
\end{theorem}

\begin{proof}
Note that $A+P_{ab}$ is non-special of degree $g$ by Theorem \ref{kummer rgtm}.
We have $A$ non-special too by Lemma \ref{lemma.BB}.
Given
$$|\Supp(A)|=r-\left\lfloor\frac{r}{m}\right\rfloor-1\leq r-1,$$
we can take $P\in\Supp((x))\setminus\Supp(A)\neq\emptyset$.
\end{proof}

\section{Hulls of codes from Kummer extensions} \label{codes_euclidean_section}

Kummer extensions include several families of well-known maximal function fields, such as the Hermitian function field. Their structure facilitates building codes with desirable properties. Self-orthogonal codes have been investigated in \cite{MTT} for the one-point codes  $C(D,mP_\infty)$. In this section, by employing the results from Section \ref{explicit_divisors_section}, we see how to get algebraic geometry codes whose hulls are algebraic geometry codes from Kummer extensions.

The general idea can be illustrated by using the non-special divisors of degree $g-1$ to get LCD codes. We consider a non-special divisor $A$ of degree $g-1$ and a divisor $B\geq A$ such that $A+B-D=(\eta)$ is canonical, with simple poles at $D$ and $\mathrm{res}_P(\eta)=1$ for $P\in\Supp(D)$. With this, we would have that any two divisors $G,H$ of degree less than $\deg D$, such that $\gcd(G,H)=A$ and $\mathrm{lcm}(G,H)=B$, satisfy
$$C(D,G)^\perp=C(D,H)\ \ \ \text{ and }\ \ \ C(D,G)\cap C(D,H)=\{0\},$$ meaning that $C(D,G)$ is LCD.

\begin{theorem}\label{maxcur}
    Let $F/\mathbb{F}_{q^2}(y)$ be a maximal Kummer extension of genus $g$ defined by 
       $\prod_{i=1}^r (y-\alpha_i)=x^{q+1}$
  where $\alpha_i\in\mathbb{F}_q$ and $(r,q+1)=1$. Take 
  $D=(x^{q^2-1}-1)_0$, meaning the sum of all places $P_{ab}$ with $a,b\in\mathbb{F}_{q^2}\setminus\{0\}$ satisfying $\prod_{i=1}^r (b-\alpha_i)=a^{q+1}$,
  and $A$ to be a non-special divisor of degree $g-1$ such that $\Supp(A)\subset\mathrm{Supp}((x)_0)$. Set $P_\infty$ to be the only pole of $F$. If $G,H$ are divisors such that $2g-2<\deg G, \deg H<\deg D$ and 
    $$\gcd(G,H)= A,$$
    $$G+H=(2g+r-2)P_\infty+\sum_{i=1}^r (q^2-2)P_{0b_i},$$ then $C(D,G)$ is LCD.
\end{theorem}

\begin{proof}
	We have $(x^{q^2-2})=(q^2-2)\sum_{i=1}^r P_{0b_i}-r(q^2-2)P_\infty$ and $(x^{q^2-1}-1)=D-r(q ^2-1)P_\infty$. Then
	$$B=\lmd(G,H)=(2g+r(q^2-1)-2)P_\infty+(x^{q^2-2})-A.$$This implies $$A+B-D=(2g-2)P_\infty-(x^{q^2-1}-1)+(x^{q^2-2}).$$
From \cite[Proposition 3]{MTT}, we know that $(dx)=(2g-2)P_\infty$, then, if $z=x^{q^2-1}-1$, we have
 $$A+B-D=(dz/z)$$which is a canonical divisor with simple poles at $D$ and residues $-1$ at any point of $D$.

    \end{proof}

\begin{remark}\rm
    Theorem \ref{maxcur} also holds for Kummer extensions defined by $Q(y)=x^m$, where $Q$ is a linearized polynomial of degree $r$, $q+1\leq m\leq \frac{q^2}{2}-(2,q)+1$ and the number of rational places of the extension is $q^2r+1$. These function fields are optimal because they attain the Lewittes bound and are maximal if and only if $m=q+1$. The norm-trace function field defined by $y^{q^{t-1}}+\cdots+y^q+y=x^\frac{q^t-1}{q-1}$ over $\mathbb{F}_{q^t}$ is such an extension, and it is maximal if and only if $t=2$. 
\end{remark}

\begin{remark}
 	Let $G=A$ and $H=(2g+r-2)P_\infty+\sum_{i=1}^r (q^2-2)P_{0b_i}-A$. By Theorem \ref{kummer rgtm}, we know that for any $R\in\Supp(A)$, $$v_R(A)\leq q-\left\lfloor\frac{q+1}{r}\right\rfloor\leq q-1.$$
        Then $v_{P_{0b_i}}(2A)\leq 2q-2\leq q^2-2$ for any $1\leq i\leq r$ and so
    $$G\leq H.$$Thus, the set of pairs $(G,H)$ satisfying the conditions of Theorem \ref{maxcur} is non-empty. Indeed, $(A,(2g+r-2)P_\infty+\sum_{i=1}^r (q^2-2)P_{0b_i}-A)$ is one such pair, though the code $C(D,A)$ is trivial. The challenge now is to find divisors satisfying these conditions to produce nontrivial codes.
\end{remark}

\medskip

Now, let us describe some LCD codes for specific function fields.
\begin{proposition}\label{function field1}
    Consider the function field $\mathbb{F}_{q^2}(x,y)/\mathbb{F}_{q^2}$ where $\sum_{i=1}^t y^{\frac{q}{2^i}}=x^{q+1}$, $q=2^t$, and $t \geq 2$.
    Let $G=(q^2-1)P_{0b_{\frac{q}{2}}}+\sum_{j=1}^{\frac{q}{2}-1}2jP_{0b_j}$
where $(x)=\sum_{i=1}^{\frac{q}{2}} P_{0b_i}-\frac{q}{2}P_\infty$, and
$D$ be as in Theorem \ref{maxcur}. Then, $C(D, G)$ is an LCD code of dimension $q^2$. 
\end{proposition}

\begin{proof}
    Take $A=\sum_{j=1}^{\frac{q}{2}-1} 2jP_{0b_j}$ and $n=\deg D=\frac{q}{2}(q^2-1)$. By Theorem \ref{kummer rgtm}, $A$ is non-special of degree $g=\frac{(q-2)q}{2}$, the genus of the function field. Take $$H=\frac{q^2-q-4}{2}P_\infty+\sum_{i=1}^{\frac{q}{2}-1} (q^2-2-2j)P_{0b_i}-P_{0b_\frac{q}{2}}.$$ Then
    $$\gcd(G,H)=A-P_{0b_\frac{q}{2}}$$
and    $$\lmd(G,H)=\frac{q^2-q-4}{2}P_\infty+\sum_{i=1}^\frac{q}{2}(q^2-2)P_{0b_i}-A+P_{0b_\frac{q}{2}}.$$
    
    By Theorem \ref{maxcur}, $C(D, G)$ is LCD. Since $A\leq G$ and $A$ is non-special, $G$ is non-special. Given that $\deg G=g+q^2-1<n$, we have
    $$\dim C(D,G)=\ell(G)=\deg G-g+1=q^2$$ and the result follows.
    \end{proof}

\begin{proposition}\label{function field2}
    Consider the function field $\mathbb{F}_{q^{2r}}(x,y)/\mathbb{F}_{q^{2r}}$ where $y^q+y=x^{q^r+1}$, $q$ is a power of a prime, and $r$ is odd. Let 
    $$G=(q^r+1)(q-1)P_\infty+\sum_{i=1}^{q-1}q^{r-1}jP_{0b_j},$$ where  $(x)=\sum_{i=1}^q P_{0b_i}-qP_\infty$, and take 
$D$ as in Theorem \ref{maxcur}. Then $C(D,G)$ is an LCD code of dimension $(q^r+1)(q-1)+1$. 
\end{proposition}

\begin{proof}
    By Theorem \ref{kummer rgtm}, $A=\sum_{i=1}^{q-1}q^{r-1}jP_{0b_j}$  is non-special of degree $g=\frac{q^r(q-1)}{2}$. Take $H=\sum_{i=1}^q (q^{2r}-2)P_{0b_i}-A-P_\infty$. We have $\gcd(G,H)=A-P_\infty$ and
    $$G+H=((q^r+1)(q-1)-1)P_\infty+\sum_{i=1}^{q}(q^{2r}-2)P_{0b_i}.$$
    
    Then, by Theorem \ref{maxcur}, $C(D, G)$ is LCD. Since $\deg G=g+(q^r+1)(q-1)<\deg D=q(q^{2r}-1)$ and $G$ is non-special, $\dim C(D,G)=\ell(G)=(q^r+1)(q-1)+1$.
\end{proof}


The following are examples of LCD codes over maximal function fields.

\begin{example}\rm
Take $q=4$ and consider the function field $F/\mathbb{F}_{q^2}$, where $y^2+y=x^5$. Then $F$ has $33$ rational places and genus $g=2$. Take $P_\infty$ to be the place at infinity and $D=\sum_{a\in\mathbb{F}_{16}\setminus\{0\}} P_{ab}$. We have that $\deg D=30$. For $G=4P_\infty+12P_{00}-P_{01}$, we have that $C(D, G)$ is of dimension $14$, minimum distance $15$, and it is generated by the evaluations of the following functions
$$\mathcal B:=\left\{x,x^2,\frac{x}{y},\frac{x^2}{y},\frac{x^3}{y},\frac{x^4}{y},\frac{x}{y^2},\frac{x^2}{y^2},\frac{x^3}{y^2},\frac{x^4}{y^2},\frac{x^5}{y^2},\frac{x^3}{y^3},\frac{x^4}{y^3},\frac{x^5}{y^3}\right\}.$$
   
On the other hand, for $H=2P_{00}+15P_{01}$, $C(D,H)$ is of dimension $16$, minimum distance $13$, and generated by the evaluation functions
$$ \left\{   \begin{array}{l}
1, \frac{1}{x^2}, \frac{1}{x}, \frac{1}{x^2(y+1)}, \frac{1}{x(y+1)}, \frac{1}{y+1}, \frac{x}{y+1}, \frac{x^2}{y+1},\frac{1}{x^2(y+1)^2},\\
\frac{1}{x(y+1)^2},\frac{1}{(y+1)^2},\frac{x}{(y+1)^2},\frac{x^2}{(y+1)^2},\frac{1}{(y+1)^3},\frac{x}{(y+1)^3},\frac{x^2}{(y+1)^3}
\end{array}\right\}.$$
    
By Proposition \ref{function field1}, $C(D,G)^\perp=C(D,H)$ and $C(D,G)+C(D,H)=\mathbb{F}_{16}^{30}$.
\end{example}

\begin{example}\rm
Let $q=2$ and consider the Hermitian function field defined by $y^2+y=x^3$ over $\mathbb{F}_4$. This function field has 9 rational places and genus $g=1$. Take $P_\infty$ the place at infinity and $D=\sum_{a\in\mathbb{F}_4\setminus\{0\}} P_{ab}$. We have $\deg D=6$ and for $G=3P_\infty+P_{00}$, $C(D,G)$ is of dimension 4 and minimum distance 2. Its generator matrix is
    $$\begin{array}{l|cccccc}
    &(a,a)&(a^2,a)&(1,a)&(a,a^2)&(a^2,a^2)&(1,a^2)\\\hline
    \frac{x^2}{y}&a&1&a^2&1&a^2&a\\
    y&a&a&a&a^2&a^2&a^2\\
    x&a&a^2&1&a&a^2&1\\
    1&1&1&1&1&1&1
    \end{array}$$
    
    Now, for $H=P_{00}+2P_{01}-P_\infty$, $C(D, H)$ is of dimension 2, minimum distance 4, and its generator matrix is
    $$\begin{array}{l|cccccc}
    &(a,a)&(a^2,a)&(1,a)&(a,a^2)&(a^2,a^2)&(1,a^2)\\\hline
    \frac{1}{x}&a^2&a&1&a^2&a&1\\
    \frac{x}{y+1}&a^2&1&a&1&a&a^2
    \end{array}.$$
We can easily check that $C(D,G)^\perp=C(D,H)$ and $C(D,G)\cap C(D,H)=\{0\}$. These codes are near MDS; the sum of their minimum distances equals their length.
\end{example}

Carlet, Mesnager, Tang, Qi, and Pellikaan proved that for $q>3$, any code is isometric to an LCD code, and finding the isometry is relatively simple since it just depends on the generator matrices of the code and some linear algebra \cite{linear}. Even more, Luo et al. used isometries of codes to reduce the hull of a code~\cite[Theorem 7]{Grassl22}. Their result is easily extended to the Euclidean hull. Here, we are not considering how to modify codes to get specific hulls, such as LCD, but rather how to build algebraic geometry codes whose hulls can be fully characterized as algebraic geometry codes using only rational points. The following proposition illustrates how the same techniques we used to describe LCD codes can be used to get the same result for bigger hulls.

\begin{corollary}
Let $G,H$ and $A$ as in Theorem \ref{maxcur}, but with $\gcd(G,H)\geq A$. Then 
$$C(D,G)\cap C(D,H)=C(D,\gcd(G,H))$$
and the dimension of $C(D,\gcd(G,H))$ is $\deg(\gcd(G,H))-g+1$.
\end{corollary}

As an example, we have codes over the Hermitian curve.

\begin{theorem}
    Consider the Hermitian function field $F$ given by $y^q+y=x^{q+1}$ over $\mathbb{F}_{q^2}$. Let
    $$G=\sum_{j=1}^{q}\alpha_jP_{0bj}+\alpha_\infty P_{\infty},$$
      $(x)=\sum_{i=1}^q P_{0b_i}-qP_\infty$, 
 and $D$ be the sum of places $P_{ab}$ on $F$ where $a,b\in\mathbb{F}_{q^2}\setminus{0}$, $b^q+b=a^{q+1}$, $j\leq\alpha_j<q^2-j-2$ for $1\leq j\leq q-1$, and $0\leq\alpha_q\leq q^2-2$ and $-1\leq\alpha_\infty\leq q^2-1$. Let $c_j=\min\{\alpha_j,q^2-\alpha_j-2\}$. Then
 $$\mathrm{Hull}(C(D,G))=C\left(D,\sum_{j=1}^q c_jP_{0b_j}+c_\infty P_\infty\right)$$

 \noindent which is a code of dimension $\sum_{i=1}^{q}c_j+c_\infty-\frac{q(q-1)}{2}+1$.
\end{theorem}

\begin{proof}
    Let $H=\sum_{i=1}^q(q^2-2)P_{0b_i}+(q^2-2)P_\infty-G$. Since $j\leq\alpha_j<q^2-2-j$ for $1\leq j\leq q-1$, then $\mathrm{gcd}(G,H)\geq \sum_{j=1}^{q-1} j P_{0b_j}+c_\infty P_\infty$, which is a non-special divisor of degree $g+c_\infty\geq g-1$ by Theorem \ref{kummer rgtm}. We have the conclusion since $G+H-D$ is a canonical divisor, as proved in the proof of Theorem \ref{maxcur}. 
\end{proof}

\begin{example}
    Some extreme cases are the following.

    \begin{itemize}
        \item If $\alpha_j=j$ for $1\leq j\leq q-1$ and $\alpha_\infty,\alpha_q\in\{0,q^2-2\}$, then $C(D,G)$ is a LCD code.

        \item If $\alpha_j=j$ for $1\leq j\leq q-1$ and $\alpha_\infty=q^2-1$, then $C(D,G+P)$ is a one-dimensional-hull code for any rational $P$.

        \item If $\alpha_j\leq q^2-2-\alpha_j$ for any $j\in\{1,\ldots,q\}\cup\{\infty\}$, then $C(D,G)$ is self-orthogonal.
    \end{itemize}
\end{example}

\begin{example}
Let $q=2$ and consider the Hermitian function field defined by $y^2+y=x^3$ over $\mathbb{F}_4$ as before. Take $D=\sum_{a\in\mathbb{F}_4\setminus\{0\}} P_{ab}$,  $G=2P_\infty+P_{01}+P_{00}$ and $H=P_{00}+P_{01}$. We have $\gcd(G,H)=P_{00}+P_{01}\geq P_{00}-P_\infty$ and then $\gcd(G,H)$ is non-special of degree $2$. Therefore, $\dim C(D,G)\cap C(D,H)=2$ and the minimum distance of $C(D, G')$ is 2. 
\end{example}

To conclude this study of hulls of algebraic geometry codes from Kummer extensions, let us finish by describing how to modify one of the divisors described in Theorem \ref{maxcur} to get more divisors that still yield codes with algebraic geometry hull. 

\begin{theorem}\label{22.12.04}
Let $G$ and $ H$ be divisors as in Theorem \ref{maxcur}. Let $S$ be a degree-zero divisor such that $v_P(S)=0$ if $v_P(G)=v_P(H)$ or 
$$0\leq\frac{v_P(S)}{v_P(H)-v_P(G)}\leq 1.$$
Then, $\mathrm{Hull}(C(D,G+S))=C(D,\gcd(G+S,H-S))$ and $C(D,G+S)$ has the same dimension as $C(D,G)$.
\end{theorem}

\begin{proof}
We have from the hypothesis that
$$|v_P(G)-v_P(H)+2v_P(S)|\leq|v_P(G)-v_P(H)|.$$
Therefore
	$$\gcd(G,H)\leq\gcd(G+S,H-S).$$
Since $\gcd(G, H)$ is non-special,  $\gcd(G+S, H-S)$ is also non-special. We have the equivalence of hull of $C(D, G)$ by Theorem \ref{thm_2_10_25}. Finally, since $G+S$ is non-special and $\deg S=0$, we have that $\dim C(D,G)=\dim C(D,G+S)$. 
\end{proof}

\begin{remark}\rm
    Not all zero divisors are principal. Thus, Theorem \ref{22.12.04} can describe different modifications than those described by Carlet and Guilley \cite{CG} and Luo, Ezerman,  Grassl, and Ling \cite{Grassl22}. Even more, our theorem describes a way to increase the hull given an LCD code, which is not guaranteed by these previous works.
\end{remark}

\section{Conclusion} \label{conclusion}

In this paper, we determined explicit non-special divisors of degree $g-1$ and $g$ on certain Kummer extensions. Consequently, we obtained the expressions for non-special divisors of the smallest degree and effective non-special divisors of the least degree on families of Hermitian and norm-trace function fields and some of their quotients. These divisors support the construction of algebraic geometry codes with algebraic geometry hulls on some families of maximal function fields (and some close relatives) using only rational places, which are relevant to coding theory. These results provide the ability to construct algebraic geometry codes with hulls of specified dimensions, including LCD codes.


\begin{thebibliography}{99}

\bibitem{Assmus}
E. Assmus and J. Key, Affine and projective planes, Discrete Math.,
83 no. 2 (1990), pp. 161--187, 1990.
%

\bibitem{BB}
S. Ballet and D. Le Brigand, On the existence of non-special divisors of degree $g$ and $g - 1$ in algebraic function fields over $\F_q$, J. Number Theory, 116 (2006), pp. 293--310. 
%

\bibitem{graph_isom}
M. Bardet, A. Otmani, and M. Saeed-Taha, Permutation code equivalence is not
harder than graph isomorphism when hulls are trivial, In 2019, IEEE International
Symposium on Information Theory (ISIT), pp. 2464--2468, 2019.
%


\bibitem{Carlet}
C. Carlet, Correlation-Immune Boolean Functions for Leakage Squeezing and Rotating S-Box Masking against Side Channel Attacks, In Gierlichs, B., Guilley, S., Mukhopadhyay, D. (eds) Security, Privacy, and Applied Cryptography Engineering. SPACE 2013. Lecture Notes in Computer Science, 8204. Springer, Berlin, Heidelberg, pp. 70--74, 2013. 

\bibitem{CG}
C. Carlet and S. Guilley, Complementary dual codes for counter-measures to side-channel attacks, Coding Theory and Applications, Springer, Cham, pp. 97--105, 2015.
%

\bibitem{linear} 
C. Carlet, S. Mesnager, C. Tang, Y. Qi and R. Pellikaan, Linear Codes Over $\mathbb F_q$ Are Equivalent to LCD Codes for $q>3$,  IEEE Trans. Inform. Theory, 64 no. 4 (2018), pp. 3010--3017.
%

\bibitem{oneandtwo}
A. S. Castellanos, A. M. Masuda and L. Quoos, One- and Two-point Codes Over Kummer Extensions, IEEE Trans. Inform. Theory, 62 no. 9 (2016), pp. 4867--4872.


\bibitem{hull_quantum}
K. Guenda, S. Jitman, and T. A. Gulliver, Constructions of good entanglement-
assisted quantum error correcting codes, Des. Codes and Cryptogr., 86 no. 1 (2018), pp. 121--136.
%

\bibitem{Grassl22} G. Luo, M. F. Ezerman, N, Grassl, \& S. Ling, How Much Entanglement Does a Quantum Code Need?, arXiv preprint version 2 (2022), arXiv:2207.05647.
%

\bibitem{Massey} J. L. Massey, Linear codes with complementary duals, Discrete Math.,106--107 (1992), pp. 337--342. 
%

\bibitem{Matthews}
G. L. Matthews, The Weierstrass semigroup of an $m$-tuple of collinear places on a Hermitian function field, in Finite Fields and Applications, Lecture Notes in Comput. Sci., 2498 (2004), pp. 12--24.
%

\bibitem{MTQ}
S. Mesnager, C. Tang, and Y. Qi, Complementary dual algebraic geometry codes,  IEEE Trans. Inform. Theory, 64 no. 4 (2018), pp. 2390--2397.
%

\bibitem{MTT}
C. Munuera, W. Ten\'orio, and F. Torres, Quantum error-correcting codes from algebraic geometry codes of Castle type, Quantum Inf. Process., 15 no. 10 (2016), pp. 4071--4088.
%



\bibitem{Adv}
Niederreiter, H., \& \"Ozbudak, F. (2008). Asymptotically good codes. In \textit{Advances in Algebraic Geometry Codes} (pp. 181-220).


\bibitem{NXL}
Niederreiter, H., Xing, C., \& Lam, K. Y. (1999). A new construction of algebraic geometry codes. Applicable Algebra in Engineering, Communication and Computing, 9, 373-381.
%



\bibitem{Peachey}
J. Peachey, Bases and applications of Riemann-Roch spaces of function fields with many rational places, [Doctoral dissertation], Clemson University, 2011.
%

\bibitem{tower}
Pellikaan, R., Stichtenoth, H., \& Torres, F. (1998). Appeared in: Finite Fields and their Applications, vol. 4, pp. 381-392, 1998. Weierstrass semigroups in an asymptotically good tower. Finite fields and their applications, 4, 381-392.
%

\bibitem{Schmidt}
F. K. Schmidt, Zur arithmetischen Theorie der algebraischen Funktionen. II. Allgemeine Theorie der Weierstra punkte, Math. Z., 45 (1939), pp.  75--96. 
%

\bibitem{Sendrier}
N. Sendrier, Finding the permutation between equivalent linear codes: the support
splitting algorithm, IEEE Trans. Inform. Theory, 46 no. 4 (2000), no. 1193--1203.
%

\bibitem{St}
H. Stichtenoth,  Algebraic Function Fields and Codes,
Springer-Verlag, 1993.

\bibitem{kummer} S. Yang and C. Hu, Weierstrass semigroups from Kummer extensions, Finite Fields and Their Appl., 45 (2017), pp. 264--284.


\bibitem{vL}
J. H. van Lint and G. van der Geer, Introduction to Coding Theory and Algebraic Geometry, 1988.
%

\end{thebibliography}
\end{document}